\documentclass{llncs}

\usepackage{amsfonts}
\usepackage{graphicx}
\usepackage{tikz}
\usepackage[noend,ruled,linesnumbered,noline]{algorithm2e}

\def\A{{\mathcal A}}
\def\H{{\mathcal H}}

\def\E{{\mathcal E}}
\def\O{{\mathcal O}}
\def\U{{\mathcal U}}
\def\F{{\mathcal F}}
\def\S{{\mathcal S}}
\def\D{{\mathcal D}}

\begin{document}
\pagestyle{plain}
\title{Counting Popular Matchings in House Allocation Problems}
\author{Rupam Acharyya, Sourav Chakraborty, Nitesh Jha}
\institute{
  Chennai Mathematical Institute\\
  Chennai, India\\
  \email{\{rupam,sourav,nj\}@cmi.ac.in}
}
\maketitle


\begin{abstract}
We study the problem of counting the number of \textit{popular matchings} in a given instance. A popular matching instance consists of agents $\A$ and houses $\H$, where each agent ranks a subset of houses according to their preferences. A \textit{matching} is an assignment of agents to houses. A matching $M$ is \textit{more popular} than matching $M'$ if the number of agents that prefer $M$ to $M'$ is more than the number of people that prefer $M'$ to $M$. A matching $M$ is called \textit{popular} if there exists no matching more popular than $M$. McDermid and Irving gave a poly-time algorithm for counting the number of popular matchings when the preference lists are strictly ordered. \\

We first consider the case of ties in preference lists. Nasre proved that the problem of counting the number of popular matching is \#P-hard when there are ties. We give an FPRAS for this problem. \\

We then consider the popular matching problem where preference lists are strictly ordered but each house has a capacity associated with it. We give a \textit{switching graph characterization} of popular matchings in this case. Such characterizations were studied earlier for the case of strictly ordered preference lists (McDermid and Irving) and for preference lists with ties (Nasre). We use our characterization to prove that counting popular matchings in capacitated case is \#P-hard.
\end{abstract}

\newpage


\section{Introduction}
\label{sec:intro}
A $\textit{popular matching problem}$ instance $I$ comprises a set $\A$ of $\textit{agents}$ and a set $\H$ of $\textit{houses}$. Each agent $a$ in $\A$ ranks (numbers) a subset of houses in $\H$ (lower rank specify higher preference). The ordered list of houses ranked by $a \in \A$ is called $a$'s \textit{preference list}. For an agent $a$, let $E_a$ be the set of pairs $(a,h)$ such that the house $h$ appears on $a$'s preference list. Define $E = \cup_{a \in \A} E_a$. The problem instance $I$ is then represented by a bipartite graph $G = (\A \cup \H, E)$. A $matching$ $M$ of $I$ is a matching of the bipartite graph $G$. We use $M(a)$ to denote the house assigned to agent $a$ in $M$ and $M(h)$ to denote the agent that is assigned house $h$ in $M$. An agent $\textit{prefers}$ a matching $M$ to a matching $M'$ if $(i)$ $a$ is matched in $M$ and unmatched in $M'$, or $(ii)$ $a$ is matched in both $M$ and $M'$ but $a$ prefers the house $M(a)$ to $M'(a)$. Let $\phi(M,M')$ denote the number of agents that prefer $M$ to $M'$. We say $M$ is $\textit{more popular than}$ $M'$ if $\phi(M,M') > \phi(M',M)$, and denote it by $M \succ M'$. A matching $M$ is called $\textit{popular}$ if there exists no matching $M'$ such that $M' \succ M$.

The popular matching problem was introduced in \cite{gardenfors1975match} as a variation of the stable marriage problem \cite{gale1962college}. The idea of popular matching has been studied extensively in various settings in recent times \cite{DBLP:journals/siamcomp/AbrahamIKM07,DBLP:journals/jda/SngM10,DBLP:journals/jco/McDermidI11,DBLP:conf/sigecom/Mahdian06,DBLP:journals/tcs/KavithaMN11,DBLP:conf/latin/McCutchen08,DBLP:conf/stacs/Nasre13}, mostly in the context where only one side has preference of the other side but the other side has no preference at all. We will also focus on this setting. Much of the earlier work focuses on finding efficient algorithms to output a popular matching, if one exists.

The problem of counting the number of ``solutions" to a combinatorial question falls into the complexity class \#P. An area of interest that has recently gathered a certain amount of attention is the problem of counting stable matchings in graphs. The Gale-Shapely algorithm \cite{gale1962college} gives a simple and efficient algorithm to output a stable matching, but counting them was proved to be \#P-hard in \cite{DBLP:journals/siamcomp/IrvingL86}. Bhatnagar, Greenberg and Randall \cite{DBLP:conf/soda/BhatnagarGR08} showed that the random walks on the \textit{stable marriage lattice} are slowly mixing, even in very restricted versions of the problem. \cite{DBLP:conf/approx/CheboluGM10} gives further evidence towards the conjecture that there may not exist an FPRAS at all for this problem.

Our motivation for this study is largely due to the similarity of structures between stable matchings and popular matchings (although no direct relationship is known). The interest is further fueled by the existence of a linear time algorithm to exactly count the number of popular matchings in the standard setting \cite{DBLP:journals/jco/McDermidI11}. We look at generalizations of the standard version - preferences with ties and houses with capacities. In the case where preferences could have ties, it is already known that the counting version is \#P-hard \cite{DBLP:conf/stacs/Nasre13}. We give an FPRAS for this problem. In the case where houses have capacities, we prove that the counting version is \#P-hard. While the FPRAS for the case of ties is achieved via a reduction to a well known algorithm, the \#P-hardness for the capacitated case is involving, making it the more interesting setting of the problem.

We now formally describe the different variants of the popular matching problem (borrowing the notation from \cite{DBLP:journals/jda/SngM10}) and also describe our results alongside.

\ \\\noindent
$\textbf{House Allocation problem (HA)}$
These are the instances $G = (\A \cup \H, E)$ where the preference list of each agent $a \in \A$ is a linear order. Let $n = |\A| + |\H|$ and $m = |E|$. In \cite{DBLP:journals/siamcomp/AbrahamIKM07}, Abraham et~al.\ give a complete characterization of popular matchings in an HA instance, using which they give an $O(m + n)$  time algorithm to check if the instance admits a popular matching and to obtain the largest such matching, if one exists. The question of counting popular matchings was first addressed in \cite{DBLP:journals/jco/McDermidI11}, where McDermid et~al.\ give a new characterization by introducing a powerful structure called the $\textit{switching graph}$ of an instance. The switching graph encodes all the popular matchings via $\textit{switching paths}$ and $\textit{switching cycles}$. Using this structure, they give a linear time algorithm to count the number of popular matchings. 

\ \\\noindent
$\textbf{House Allocation problem with Ties (HAT)}$ An instance $G = (\A \cup \H, E)$ of HAT can have applicants whose preference list contains ties. For example, the preference list of an agent could be $[h_3, (h_1, h_4), h_2]$, meaning, house $h_3$ gets rank 1, houses $h_1$ and $h_4$ get a tied rank 2 and house $h_2$ gets the rank 3. A characterization for popular matchings in HAT was given in \cite{DBLP:journals/siamcomp/AbrahamIKM07}. They use their characterization to give an $O(\sqrt{n}m)$ time algorithm to solve the maximum cardinality popular matching problem. We outline their characterization briefly in Section \ref{sec:HAT} where we consider the problem of counting popular matchings in HAT. In \cite{DBLP:conf/stacs/Nasre13}, Nasre gives a proof of $\#$P-hardness of this problem. We give an FPRAS for this problem by reducing it to the problem of counting perfect matchings in a bipartite graph. 

\ \\\noindent
$\textbf{Capacitated House Allocation problem (CHA)}$ A popular matching instance in CHA has a \textit{capacity} $c_i$ associated with each house $h_i \in \H$, allowing at most $c_i$ agents to be matched to house $h_i$. The preference list of each agent is strictly ordered. A characterization for popular matchings in CHA was given in \cite{DBLP:journals/jda/SngM10}, along with an algorithm to find the largest popular matching (if one exists) in time $O(\sqrt{C}n_1 + m)$, where $n_1 = |\A|$, $m = |E|$ and $C$ is the total capacity of the houses. In Section $\ref{sec:CHA}$, we consider the problem counting popular matchings in CHA. We give a switching graph characterization of popular matchings in CHA. This is similar to the switching graph characterization for HA in \cite{DBLP:journals/jco/McDermidI11}. Our construction is also motivated from  \cite{DBLP:conf/stacs/Nasre13}, which gives a switching graph characterization of HAT. We use our characterization to prove that it is \#P-Complete to compute the number of popular matchings in CHA. 

\ \\\noindent
$\textbf{Remark:}$ A natural reduction exists from a CHA instance $G = (\A \cup \H, E)$ to an HAT instance. The reduction is as follows. Treat each house $h_i \in \H$ with capacity $c$ as $c$ different houses $h_i^1, \ldots, h_i^c$ of unit capacity, which are always tied together and appear together wherever $h_i$ appears in any agent's preference list. Let the HAT instance thus obtained be $G'$. It is clear that every popular matching of $G$ is a popular matching of $G'$. Hence, for example, an algorithm which finds a maximum cardinality popular matching for HAT can be used to find a maximum cardinality popular matching for the CHA instance $G$. In the context of counting, it is important to note that one popular matching of $G$ may translate to many popular matchings in $G'$. It is not clear if there is a useful map between these two sets that may help in obtaining either hardness or algorithmic results for counting problems.


\section{Counting in House Allocation problem with Ties}
\label{sec:HAT}
In this section we consider the problem of counting the number of popular matchings in House Allocation problem with Ties (HAT). We first describe the characterization given in \cite{DBLP:journals/siamcomp/AbrahamIKM07} here using similar notations. Let $G = (\A \cup \H, E)$ be an HAT instance. For any agent $a \in \A$, let $f(a)$ define the set of first choices of $a$. For any house $h \in \H$, define $f(h) := \{a \in \A, f(a) = h \}$. A house $h$ for which $f(h) \neq \phi$ is called an $f$-house. To simplify the definitions, we add a unique last-resort house $l(a)$ with lowest priority for each agent $a \in A$. This forces every popular matching to be an applicant complete matching.

\begin{definition}
The \textbf{first choice graph} of $G$ is defined to be $G_1 = (\A \cup \H, E_1)$, where $E_1$ is the set of all rank one edges.
\end{definition}

\begin{lemma}
\label{pop-match-charac-1}
If $M$ is a popular matching of $G$, then $M \cap E_1$ is a maximum matching of $G_1$. 
\end{lemma}

Let $M_1$ be any maximum matching of $G_1$. The matching $M_1$ can be used to identify the houses $h$ that are always matched to an agent in the set $f(h)$. In this direction, we observe that $M_1$ defines a partition of the vertices $\A \cup \H$ into three disjoint sets - $\textit{even}$, $\textit{odd}$ and $\textit{unreachable}$: a vertex is $\textit{even}$ (resp.\ $\textit{odd}$) if there is an even (resp.\ odd) length alternating path from an unmatched vertex (with respect to $M_1$) to $v$; a vertex $v$ is $\textit{unreachable}$ if there is no alternating path from an unmatched vertex to $v$. Denote the sets $\textit{even}$, $\textit{odd}$ and $\textit{unreachable}$ by $\E$, $\O$ and $\U$ respectively. The following is a well-known theorem in matching theory \cite{MR859549}.

\begin{lemma}[Gallai-Edmonds Decomposition]
\label{lemma:gallai-edmonds}
Let $G_1$ and $M_1$ define the partition $\E$, $\O$ and $\U$ as above. Then, 
\begin{enumerate}
\item [(a)] The sets $\E$, $\O$ and $\U$ are pairwise disjoint, and every maximum matching in $G_1$ partitions the vertices of $G_1$ into the same partition of even, odd and unreachable vertices.
\item [(b)] In any maximum matching of $G_1$, every vertex in $\U$ is matched with another vertex in $\U$, and every vertex in $\O$ is matched with some vertex in $\E$. No maximum matching contains an edge between a vertex in $\O$ and a vertex in $\O \cup \U$. The size of a maximum matching is $|\O| + |\U|/2$.
\item [(c)] $G_1$ contains no edge connecting a vertex in $\E$ with a vertex in $\U$.
\end{enumerate}
\end{lemma}

We show the decomposition of $G_1$ in Figure $\ref{fig:gallai-1}$, where we look at the bipartitions of $\U$, $\O$, and $\E$. Since $G_1$ only contained edges resulting from first-choices, every house in $\U_r$ and $\O_r$ is an $f$-house. From Lemma $\ref{lemma:gallai-edmonds}$, each such house $h \in \U_r \cup \O_r$ is matched with an agent in $f(h)$ in every maximum matching of $G_1$, and correspondingly in every popular matching of $G$ (Lemma $\ref{pop-match-charac-1}$).

\begin{figure}
\centering

\begin{tikzpicture} [x=0.7cm,y=0.7cm]
\tikzstyle{every node}=[font=\normalfont]
\draw [help lines, gray!0] (1cm,1) grid (5,8);

\node at (0.5,7.7) {\large{$\A$}};
\node at (4.5,7.7) {\large{$\H$}};

\draw (0,6) rectangle (1,7);
\node at (0.5,6.5) {$\U_l$};
\draw (4,6) rectangle (5,7);
\node at (4.5,6.5) {$\U_r$};
\draw (1.2, 6) -- (3.8, 6);
\draw (1.2, 7) -- (3.8, 7);
\node at (2.5,6.6) {$\vdots$};

\draw (0,4.5) rectangle (1,5.5);
\node at (0.5,5) {$\O_l$};
\draw (4,3.5) rectangle (5,5.5);
\node at (4.5,4.5) {$\E_r$};
\draw (1.2, 5.5) -- (3.8, 5.5);
\draw (1.2, 4.5) -- (3.8, 4.5);
\node at (2.5,5.1) {$\vdots$};

\draw (0,1) rectangle (1,3);
\node at (0.5,2) {$\E_l$};
\draw (4,2) rectangle (5,3);
\node at (4.5,2.5) {$\O_r$};
\draw (1.2,3) -- (3.8,3);
\draw (1.2,2) -- (3.8,2);
\node at (2.5,2.6) {$\vdots$};

\end{tikzpicture}

\caption{Gallai-Edmonds decomposition of the first-choice graph of $G$}
\label{fig:gallai-1}
\end{figure}
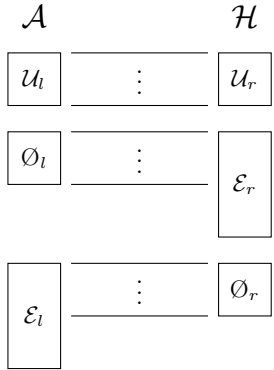

For each agent $a$, define $s(a)$ to be $a$'s most preferred house(s) in $\E_r$. Note that $s(a)$ always exists after the inclusion of last-resort houses $l(a)$. The following is proved in \cite{DBLP:journals/siamcomp/AbrahamIKM07}.

\begin{lemma}
\label{pop-match-charac-2}
A matching $M$ is popular in $G$ if and only if
\begin{enumerate}
\item $M \cap E_1$ is a maximum matching of  $G_1$, and
\item for each applicant $a$, $M(a) \in f(a) \cup s(a)$.
\end{enumerate}
\end{lemma}

\noindent
The following hardness result is from \cite{DBLP:conf/stacs/Nasre13}.

\begin{theorem}
Counting the number of popular matchings in HAT is \#P-hard.
\end{theorem}

We now give an FPRAS for counting the number of popular matchings in the case of ties. As before, let $G = (\A \cup \H, E)$ be our HAT instance. We assume that that $G$ admits at least one popular matching (this can be tested using the characterization). We reduce our problem to the problem of counting perfect matchings in a bipartite graph. We start with the first-choice graph $G_1$ of $G$, and perform a Gallai-Edmonds decomposition of $G_1$ using any maximum matching of $G_1$. In order to get a perfect matching instance, we extend the structure obtained from Gallai-Edmonds decomposition described in Figure $\ref{fig:gallai-1}$. Let $\F$ be the set of $f$-houses and $\S$ be the set of $s$-houses. We make use of the following observations in the decomposition.
\begin{itemize}
\item[---] Every agent in $\U_l$ and $\O_l$ gets one of their first-choice houses in every popular matching.
\item[---] $\E_r$ can be further partitioned into the following sets:
	\begin{itemize}
	\item[--] $\E^f_r := \{h \in \F \cap  \overline{\S}, h \in \E_r\}$,
	\item[--] $\E^s_r := \{h \in \overline{\F} \cap \S, h \in \E_r\}$,
	\item[--] $\E^{f/s}_r := \{h \in \F \cap \S, h \in \E_r\}$, and
	\item[--] $\E^\star_r := \{h \notin \F \cup \S, h \in \E_r\}$.
	\end{itemize}
\item[---] $O_l$ can only match with houses in $\E^f_r \cup \E^{f/s}_r$ in every popular matching.
\end{itemize}
These observations are described in Figure $\ref{fig:gallai-2}$(a).

\ \\
Next, we observe that every agent in $\E_l$ that is already not matched to a house in $\O_r$, must match to a house in $\E^s_r \cup \E^{f/s}_r$. We facilitate this by adding all edges $(a, s(a))$ for each agent in $\E_l$. Finally, we add a set of dummy agent vertices $\D$ on the left side to balance the bipartition. The size of $\D$ is $|\A| - (|\H| - |\E^\star_r|)$. This difference is non-negative as long as the preference-lists of agents are complete. We make the bipartition $(\D, \E^f_r \cup \E^{f/s}_r \cup \E^s_r)$ a complete bipartite graph by adding the appropriate edges. This allows us to move from one popular matching to another by switching between first and second-choices and, among second choices of agents. Finally, we remove set $\E^\star_r$ from the right side. The new structure is described in Figure $\ref{fig:gallai-2}$(b). Denote the new graph by $G'$.

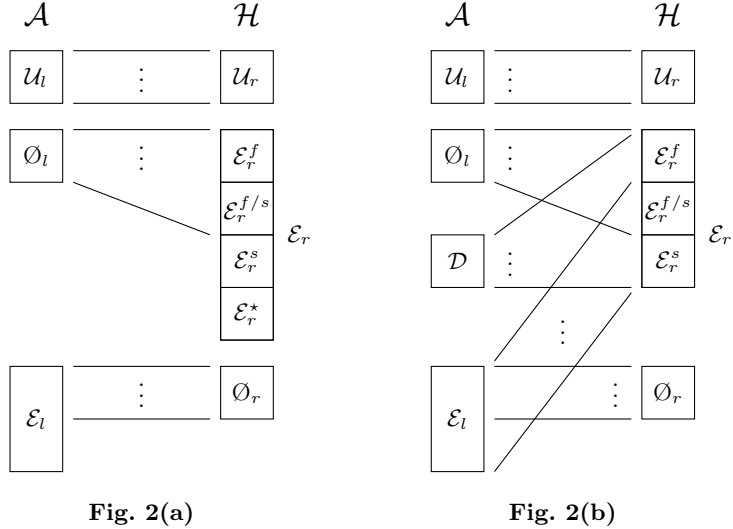
\begin{figure}
\centering

\begin{tikzpicture}[x=0.7cm,y=0.7cm] 
\tikzstyle{every node}=[font=\normalfont]
\draw [help lines, gray!0] (0,-2) grid (14,8);


\node at (0.5,7.7) {\large{$\A$}};
\node at (4.5,7.7) {\large{$\H$}};

\draw (0,6) rectangle (1,7);
\node at (0.5,6.5) {$\U_l$};
\draw (4,6) rectangle (5,7);
\node at (4.5,6.5) {$\U_r$};
\draw (1.2, 6) -- (3.8, 6);
\draw (1.2, 7) -- (3.8, 7);
\node at (2.5,6.6) {$\vdots$};

\draw (0,4.5) rectangle (1,5.5);
\node at (0.5,5) {$\O_l$};
\draw (4,1.5) rectangle (5,5.5);
\node at (5.5,3.5) {$\E_r$};
\draw (4,4.5) rectangle (5,5.5);
\node at (4.5,5) {$\E^f_r$};
\draw (4,3.5) rectangle (5,4.5);
\node at (4.5,4) {$\E^{f/s}_r$};
\draw (4,2.5) rectangle (5,3.5);
\node at (4.5,3) {$\E^s_r$};
\draw (4,1.5) rectangle (5,2.5);
\node at (4.5,2) {$\E^\star_r$};

\draw (1.2, 5.5) -- (3.8, 5.5);
\draw (1.2, 4.5) -- (3.8, 3.5);
\node at (2.5,5.1) {$\vdots$};

\draw (0,-1) rectangle (1,1);
\node at (0.5,0) {$\E_l$};
\draw (4,0) rectangle (5,1);
\node at (4.5,0.5) {$\O_r$};
\draw (1.2,1) -- (3.8,1);
\draw (1.2,0) -- (3.8,0);
\node at (2.5,0.6) {$\vdots$};

\node at (2.5,-1.8) {$\textbf{Fig. 2(a)}$};


\node at (8.5,7.7) {\large{$\A$}};
\node at (12.5,7.7) {\large{$\H$}};

\draw (8,6) rectangle (9,7);
\node at (8.5,6.5) {$\U_l$};
\draw (12,6) rectangle (13,7);
\node at (12.5,6.5) {$\U_r$};
\draw (9.2, 6) -- (11.8, 6);
\draw (9.2, 7) -- (11.8, 7);
\node at (9.5,6.6) {$\vdots$};

\draw (8,4.5) rectangle (9,5.5);
\node at (8.5,5) {$\O_l$};
\draw (12,2.5) rectangle (13,5.5);
\node at (13.5,3.5) {$\E_r$};
\draw (12,4.5) rectangle (13,5.5);
\node at (12.5,5) {$\E^f_r$};
\draw (12,3.5) rectangle (13,4.5);
\node at (12.5,4) {$\E^{f/s}_r$};
\draw (12,2.5) rectangle (13,3.5);
\node at (12.5,3) {$\E^s_r$};

\draw (9.2, 5.5) -- (11.8, 5.5);
\draw (9.2, 4.5) -- (11.8, 3.5);
\node at (9.5,5.1) {$\vdots$};

\draw (8,-1) rectangle (9,1);
\node at (8.5,0) {$\E_l$};
\draw (12,0) rectangle (13,1);
\node at (12.5,0.5) {$\O_r$};
\draw (9.2,1) -- (11.8,1);
\draw (9.2,0) -- (11.8,0);
\node at (11.5,0.6) {$\vdots$};

\draw (9.2,1.1) -- (11.8,4.5);
\draw (9.2,-1) -- (11.8,2.4);
\node at (10.5,1.8) {$\vdots$};

\draw (8,2.5) rectangle (9,3.5);
\node at (8.5,3) {$\D$};
\draw (9.2,2.5) -- (11.8,2.5);
\draw (9.2,3.5) -- (11.8,5.4);
\node at (9.5,3.1) {$\vdots$};

\node at (10.5,-1.8) {$\textbf{Fig. 2(b)}$};
\end{tikzpicture}

\caption{Reduction to a perfect-matching instance by extending the Gallai-Edmonds decomposition of $G_1$.}
\label{fig:gallai-2}
\end{figure}

\begin{lemma}
\label{lemma:per-match-red}
The number of popular matchings in $G$ is $|D|!$ times the number of perfect matchings in $G'$.
\end{lemma}

\begin{proof}
Consider a perfect matching $M$ of $G'$. Let the matching $M'$ be obtained by removing from $M$ all the edges coming out of the set $\D$. Observe that $M' \cap E_1$ is a maximum matching of $G_1$. This is because the sets $\U_l$, $O_l$ and $\O_r$ are always matched in $M'$ (or else $M$ would not be a perfect matching of $G'$) and that the size of a maximum matching in $G_1$ is $(|\U_l| + |\O_l| + |\O_r|)$ by Lemma \ref{lemma:gallai-edmonds}. Also, each agent in $\A$ is matched to either a house in $\F$ or in $\S$ by the construction of graph $G'$. Using Lemma \ref{pop-match-charac-2}, we conclude that $M$ is a popular matching of $G$. Finally, observe that every popular matching in $M$ in $G$ can be augmented to a perfect matching of $G'$ by adding exactly $|\D|$ edges. This follows again from Lemma \ref{lemma:gallai-edmonds} and Lemma \ref{pop-match-charac-2}. \qed
\end{proof}

\noindent
We now make use of the following result of Jerrum et~al.\ from \cite{DBLP:conf/stoc/JerrumSV01}.

\begin{lemma}
\label{lemma:jerrum-sinclair}
There exists an FPRAS for the problem of counting number of perfect matchings a bipartite graph.
\end{lemma}

\noindent
From Lemma $\ref{lemma:per-match-red}$ and Lemma $\ref{lemma:jerrum-sinclair}$, we have the following.

\begin{theorem}
There exists an FPRAS for counting the number of popular matchings in the House Allocation problem with Ties.
\end{theorem}

%


\section{Counting in Capacitated House Allocation problem}
\label{sec:CHA}

In this section, we consider the structure of popular matchings in Capacitated House Allocation problem (CHA). A CHA instance $I$ consists of agents $\A$ and houses $\H$. Let $|\A| = n$ and $|\H| = m$. Let $c:\H \to \mathbb{Z}_{>0}$ be the capacity function for houses. Each agent orders a subset of the houses in a strict order creating its $\textit{preference list}$.  The preference list of $a_i \in \A$ defines a set of edges $E_i$ from $a_i$ to houses in $\H$. Define $E = \cup_{i \in [n]} E_i$. The problem instance $I$ can then be represented by a bipartite graph $G = (\A \cup \H, E)$.

For the instance $I$, a $\textit{matching}$ $M$ is a subset of $E$ such that each agent appears in at most one edge in $M$ and each house $h$ appears in at most $c(h)$ edges in $M$. The definitions of $\textit{more popular than}$ relationship between two matchings and $\textit{popular matching}$ is same as described earlier in Section $\ref{sec:intro}$.

We now outline a characterization of popular matchings in CHA from \cite{DBLP:journals/jda/SngM10}. As before, denote by $f(a)$ the first choice of an agent $a \in \A$. A house which is the first choice of at least one agent is called an $f$-house. For each house $h \in \H$, define $f(h) = \{a \in \A, f(a) = h\}$. For each agent $a \in \A$, we add a unique last-resort house $l(a)$ with least priority and capacity 1.

\begin{lemma}
\label{lemma:f-hosue-max-match}
If $M$ is a popular matching then for each $f$-house $h$, $|M(h) \cap f(h)| = \textnormal{min} \{c(h), |f(h)|\}$.
\end{lemma}

\noindent
For each agent $a \in \A$, define $s(a)$ to be the highest ranked house $h$ on $a$'s preference list such that one of the following is true:
\begin{itemize}
\item $h$ is not an $f$-house, or,
\item $h$ is an $f$-house but $h \neq f(a)$ and $|f(h)| < c(h)$.
\end{itemize}

\noindent
Notice that $s(a)$ always exists after the inclusion of last-resort houses $l(a)$. The following lemma gives the characterization of popular matchings in $G$. 

\begin{lemma}
\label{lemma:CHA-charac}
A matching $M$ is popular if and only if
\begin{enumerate}
\item for every $f$-house $h \in \H$,
	\begin{itemize}
	\item[--] if $|f(h)| \le c(h)$, then every agent in $f(h)$ is matched to the house $h$, 
	\item[--] else, house $h$ is matched to exactly $c(h)$ agents, all belonging to $f(h)$,
	\end{itemize}
\item $M$ is an agent complete matching such that for each agent $a \in \A$, $M(a) \in \{f(a), s(a)\}$. 
\end{enumerate}
\end{lemma}

\subsection{Switching Graph Characterization of CHA}
We now give a $\textit{switching graph}$ characterization of popular matchings for instances from this class. Our results are motivated from similar characterizations for HA in \cite{DBLP:journals/jco/McDermidI11} and for HAT in \cite{DBLP:conf/stacs/Nasre13}. A switching graph for an instance allows us to move from one popular matching to another by making well defined walks on the switching graph. 

Consider a popular matching $M$ of an instance $G$ of CHA. The switching graph of $G$ with respect to $M$ is a directed weighted graph $G_M = (\H, E_M)$, with the edge set $E_M$ defined as follows. For every agent $a \in \A$,
\begin{itemize}
\item add a directed edge from M(a) to $\{f(a),s(a)\} \setminus M(a)$,
\item if $M(a) = f(a)$, assign a weight of $-1$ on this edge, otherwise assign a weight of $+1$. 
\end{itemize}

\noindent
Associated with the switching graph $G_M$, we have an $\textit{unsaturation degree}$  function $u_M : \H \to \mathbb{Z}_{\ge 0}$, defined $u_M(h) = c(h) - |M(h)|$. A vertex $h$ is called $\textit{saturated}$ if its unsaturation degree is 0, i.e.\ $u_M(h) = 0$. If $u_M(h) > 0$, $h$ is called $\textit{unsaturated}$. We make use of the following terminology in the foregoing discussion. We now describe some useful properties of the switching graph $G_M$.

\begin{enumerate}
\item[$\rhd$] 
\textit{Property 1:} Each vertex $h$ can have out-degree at most $c(h)$. \\
\textit{Proof.} Each edge is from a matched house to an unmatched house and since the house $h$ has a maximum capacity $c(h)$, it can only get matched to at most $c(h)$ agents. \qed \vspace{3mm}

\item[$\rhd$] 
\textit{Property 2:} 
Let $M$ and $M'$ be two different popular matchings in $G$ and let $G_{M}$ and $G_{M'}$ denote the switching graphs respectively. For any vertex house $h$, the number of $-1$ outgoing edges from $h$ is invariant across $G_{M}$ and $G_{M'}$. The number of $+1$ incoming edges on $h$ is also invariant across $G_M$ and $G_{M'}$. \\
\textit{Proof.} 
From Lemma $\ref{lemma:f-hosue-max-match}$, in any popular matching, each $f$-house $h$ is matched to exactly $min\{|f(h)|, c(h)\}$ agents and this is also the number of outgoing edges with weight $-1$. A similar argument can be made for $+1$ weighted incoming edges. \qed \vspace{3mm}

\item[$\rhd$] 
\textit{Property 3:} No $+1$ weighted edge can end at an unsaturated vertex. \\
\textit{Proof.} If a $+1$ weighted edge is incident on a vertex $h$, this means that the house $h$ is an $f$-house for some agent $a$ that is still not matched to it in $M$. But if $h$ is unsaturated then it still has some unused capacity. The matching $M'$ obtained by just promoting $a$ to $h$ is popular than $M$, which is a contradiction. \qed \vspace{3mm}

\item[$\rhd$]
\textit{Property 4:} There can be no incoming $-1$ weighted edge on a saturated vertex if all its outgoing edges have weight $-1$. \\
\textit{Proof.} A $-1$ weighted edge on a vertex $h$ implies that the house $h$ is an $s$-house for some agent $a$. But if $h$ is saturated with all outgoing edges having a  weight of $-1$, then all the capacity of $h$ has been used up by agents who had $h$ as their first choice. But by definition, $h$ can not be an $s$-house for any other agent. \qed \vspace{3mm}

\item[$\rhd$]
\textit{Property 5:} For a given vertex $h$, if there exists at least one $+1$ weighted incoming edge, then all outgoing edges are of weight $-1$ and there can be no $-1$ weighted incoming edge on $h$. \\
\textit{Proof.} Let agent $a$ correspond to any $+1$ weighted incoming edge. Suppose $h$ has an outgoing $+1$ edge ending at a vertex $h'$ and agent $a'$ corresponds to this edge. We can promote agents $a$ and $a'$ to their first choices and demote any agent which is assigned house $h'$. This leads to a matching popular than $M$. Hence all outgoing edges from $h$ must be of weight $-1$. Further, Property 3 and Property 4 together imply that there can be no incoming edge on $h$ of weight $-1$. \qed 

\end{enumerate}

\subsubsection*{Switching Moves}
We now describe the operation on the switching graph which takes us from one popular matching to another. We make use of the following terminology with reference to the switching graph $G_M$. Note that the term ``path''  (``cycle'') implies a ``directed path'' (``directed cycle''). A ``$+1$ edge''(``$-1$ edge'') means an ``edge with weight $+1$'' (``edge with weight $-1$'').
\begin{itemize}
\item A path is called an \textit{alternating path} if it starts with a $+1$ edge, ends at a $-1$ edge and alternates between $+1$ and $-1$ edges.
\item A \textit{switching path} is an alternating path that ends at an unsaturated vertex.
\item A \textit{switching cycle} is an even length cycle of alternating $-1$ and $+1$ weighted edges.
\item A \textit{switching set} is a union of edge-disjoint switching cycles and switching paths, such that at most $k$ switching paths end a vertex of unsaturation degree $k$.
\item A \textit{switching move} is an operation on $G_M$ by a switching set $S$ in which, for every edge $e$ in $S$, we reversed the direction of $e$ and flip the weight of $e$ ($+1 \leftrightarrow -1$).
\end{itemize}

\noindent
Observe that every \textit{valid} switching graph inherently implies a matching (in the context of CHA) of $G$.

Let $G_M = (\H, E_M)$ and $G_{M'} = (\H, E_{M'})$ be the switching graphs associated with popular matchings $M$ and $M'$ of the CHA instance $G = (\A \cup \H, E)$. Observe that the underlying undirected graph of $G_M$ and $G_{M'}$ are same. We have the following.

\begin{theorem}
\label{thm:switch-main}
Let $S$ be the set of edges in $G_M$ that get reversed in $G_{M'}$. Then, $S$ is a switching set for $G_M$.
\end{theorem}
We prove this algorithmically in stages.

\begin{lemma}
\label{lemma:dir-cycle}
Every directed cycle in $S$ is a switching cycle of $G_M$.
\end{lemma}
\begin{proof}
Let $C$ be any cycle in $S$. From Property 5 of switching graphs, we know that no vertex in $C$ can have an incoming edge and an outgoing edge of same weight $+1$. Similarly, since $S$ is the set of edges in $G_M$ which have opposite directions and opposite weights in $G_{M'}$, we observe that $S$ can not contain any vertex with incoming and outgoing edges both having weight $-1$ (again from Property 5). This forces the weights of cycle $C$ to alternate between $+1$ and $-1$. Moreover, this alternation forces the cycle to be of even length.
\end{proof}

\noindent
At this stage we apply the following algorithm to the set $S$.  \\

\noindent
\texttt{Reduction}$(S)$: \\
1. \texttt{while} (\textit{there exists a switching cycle $C$ in $S$}): \\
\hspace*{1mm} \textit{let $S := S \setminus C$} \\
2. \texttt{while} (\textit{$S$ is non-empty}): \\
\hspace*{1mm} $(a)$ \textit{find a longest path $P$ in $S$ which alternates between weights $+1$ and $-1$} \\
\hspace*{1mm} $(b)$ \textit{let $S := S \setminus P$} \\

\noindent
At the end of every iteration of the \texttt{while} loop in Step 1, Lemma $\ref{lemma:dir-cycle}$ still holds true. We now prove a very crucial invariant of the while loop in Step 2.

\begin{lemma}
In every iteration of the \textnormal{\texttt{while}} loop in Step 2 of the algorithm \textnormal{\texttt{Reduction}}, the longest path in step $2(a)$ is a switching path for $G_M$.
\end{lemma}
\begin{proof}
Let us denote the stages of the run of algorithm \texttt{Reduction} by $t$. Initially, at $t = 0$, before any of the while loops run, $S$ is exactly the difference of edges in $E_M$ and $E_M'$. Let the while loop in Step 1 runs $t_1$ times and the while loop in Step 2 runs $t_2$ times.

Let the current stage be $t = t_1 + i$. Let $P$ be the maximal path in step $2(a)$ at this stage. We show that $P$ starts with an edge of weight $+1$. For contradiction, let $(h_i,h_j)$ be an edge of weight $-1$ and that this is the first edge of path $P$. Let $a_{ij}$ be the agent associated with the edge $(h_i, h_j)$.

The Property 5 of switching sets precludes any incoming edge of weight $-1$ on the vertex $h_i$. Hence, no switching path could have ended at $h_i$ at any stage $t < t_1 + i$. Similarly, no switching cycle with an incoming edge $-1$ was incident on $h_i$ at an earlier stage. 

Let us assume that there were $r$ cycles that were incident at $h_i$ at $t = 0$. At stage $t = t_1 + i$, let the number of outgoing $-1$ edges be $m$. Hence at $t = 0$, $h_i$ had $r$ incoming $+1$ edges and $r + m$ outgoing $-1$ edges. But this would also imply that at $t = 0$, $h_i$ had $r + m$ incoming $+1$ edges in $G_{M'}$. This contradicts Property 2, requiring the number of incoming $+1$ edges to be constant in the switching graphs corresponding to different popular matchings.

A similar argument can be made for the fact that the path $P$ can only end at an edge with weight $-1$ and that $P$ ends at an unsaturated vertex.
\end{proof}

\noindent
The following theorem establishes the characterization for popular matchings in CHA.

\begin{theorem}
\label{thm:final-switch}
If $G_M$ is the switching graph of the CHA instance $G$ with respect to a popular matching $M$, then \\
(i) every switching move on $G_M$ generates another popular matching, and \\
(ii) every popular matching of $G$ can be generated by a switching move on $M$. 
\end{theorem}

\begin{proof} \noindent
\begin{enumerate}
\item[(i)] We verify that the new matching generated by applying a switching move on $G_M$ satisfies the characterization in Lemma $\ref{lemma:CHA-charac}$. Call the new switching graph $G_{M'}$ and the associated matching $M'$. First, observe that $M'$ is indeed an agent complete matching since $G_{M'}$ still has a directed edge for each agent in $\A$. Next, each agent $a$ is still matched to $f(a)$ or $s(a)$ as the switching move either reverses an edge of $G_M$ or leaves it as it is. Finally, for each house $h$, $f(h) \subseteq M'(h)$ if $|f(h)| < c(h)$ and $|M'(h)| = c(h)$ with $M'(h) \subseteq f(h)$ otherwise. This is true because $|M'(h)| = |M(h)|$, by the definition of switching moves.

\item[(ii)] This is implied by Theorem $\ref{thm:switch-main}$.

\end{enumerate}
\end{proof}

\subsection{Hardness of Counting}
In this section we prove the \#P-hardness of counting popular matchings in CHA. We reduce the problem of counting the number of matchings in a bipartite graph to our problem.

Let $G = (A \cup B, E)$ be a bipartite matching instance in which we want to count the number of matchings. From $G$ we create a CHA instance $I$ such that the number of popular matchings of $I$ is same as the number of matchings of $G$.

Observe that a description of a switching graph gives the following information about its instance:
\begin{itemize}
\item the set of agents $A$,
\item for each agent $a \in \A$, it gives $f(a)$ and $s(a)$, and
\item for each $s$-house or $f$-house $h$, the unsaturation degree gives the capacity $c(h)$.
\end{itemize}

\noindent
Using this information, we can create the description of the instance $I$ so that it meets our requirement. For simplicity, we assume $G$ to be connected (as isolated vertices do not affect the count). We orient all the edges of $G$ from $A$ to $B$ and call the directed graph $G' = (A \cup B, E')$. Using $G'$, we construct a graph $S$, which will be the switching graph.

Let $|A| = n_1$, $|B| = n_2$ and $|E'| = m$. $S$ is constructed by augmenting $G'$. We keep all the vertices and edges of $G'$ in $S$ and assign each edge a weight of $-1$. Further, for each vertex $u \in A$, add a copy $u'$ and add a directed edge from $u'$ to $u$, and assign a weight of $+1$ to the edge. Call the new set of vertices $A'$.  The sets $A'$ and $B$ contain $s$-houses and the set $A$ contains $f$-houses. We label every vertex in $A'$ and $A$ as \textit{saturated} and for each vertex $v$ in $B$, we label $v$ as \textit{unsaturated} with \textit{unsaturation degree} 1. Hence, the switching graph $S$ has $2 n_1 + n_2$ vertices and $n_1 + m$ edges.

The CHA instance $I$ corresponding to the switching graph $S$ has $2 n_1 + n_2$ houses and $n_1 + m$ agents. Each agent has a preference list of length $2$ that is naturally defined by the weight of edges in $S$.

Let the popular matching represented by $S$ be $M_\phi$. This corresponds to the empty matching of $G$. Every non-empty matching of $G$ can be obtained by a switching move on $S$. We make this more explicit in the following theorem.

\begin{theorem}
The number of matchings in $G$ is same as the number of popular matchings in $I$.
\end{theorem}
\begin{proof}
We prove this by showing that each matching in $G$ corresponds to a unique set of edge disjoint switching paths in the switching graph $S$ of $I$.

Consider a matching $M$ of $G$ and let $(u,v) \in M$. We look at the length $2$ directed path in $S$ that is obtained by extending $(u,v)$ in the reverse direction: $u' \rightarrow u \rightarrow v$ with $u' \in A'$. It's easy to see that this is a switching path for $I$. Moreover, the set of switching paths obtained from any matching of $G$ forms a valid switching set (as every pair of such paths arising from a matching are always edge disjoint). 

For the converse, observe that $S$ can only have switching paths of length $2$ and it has no switching cycles. An edge disjoint set of such paths corresponds to a matching of $G$. By the definition of $S$, it's easy to see every matching in $M$ can be obtained by a switching set of $S$.
\end{proof}


\ \\\noindent
\textit{Conclusions and Acknowledgement:}
Our main contribution is the \#-P hardness and an FPRAS  for the Capacitated House Allocation problem. We believe that the switching graph characterization may be useful in other problems in the setting of CHA.

We thank Meghana Nasre for fruitful discussions. We also thank anonymous reviewers for their input.
 \\



\bibliographystyle{plain}
\bibliography{popular}

\end{document}